\documentclass[copyright]{eptcs}
\providecommand{\event}{DCFS 2010} 
\usepackage{breakurl}             

\usepackage{latexsym}
\usepackage[dvips]{graphics}
\usepackage{epsfig}
\usepackage{graphicx}
\usepackage{amssymb}
\begin{document}
\def\bbbn{{\rm I\!N}} 
\newcommand{\latex}{\LaTeX}
\newcommand{\tex}{\TeX}
\newcommand{\set}[1]{\{#1\}}
\newcommand{\flor}[1]{\lfloor{#1}\rfloor}

\def \prend{\vrule depth-1pt height7pt width6pt}
\def \proof{\bigbreak\noindent{\bf Proof.\ \ }}
\def \endpf{{\ \ \prend \medbreak}}

\def\vs{\vspace{1mm}}

\newtheorem{theorem}{T\/heorem}[section]
\newtheorem{apptheo}{T\/heorem}[section]
\newtheorem{corollary}{Corollary}[section]
\newtheorem{definition}{Definition}[section]
\newtheorem{lemma}{Lemma}[section]
\newtheorem{applem}{Lemma}[section]
\newtheorem{example}{Example}[section]
    \newtheorem{appexample}{Example}[section]
\newtheorem{fact}{Fact}[section]
\newtheorem{claim}{Claim}[section]
\newtheorem{proposition}{Proposition}
\newtheorem{remark}{Remark}[section]
\newcommand{\propersubset}{\subset}
\newtheorem{problem}{Open problem}

%
     %
\title{Operational State Complexity of\\
  Deterministic Unranked Tree Automata}

\author{Xiaoxue Piao \qquad\qquad Kai Salomaa
\institute{School of Computing, Queen's University\\
Kingston, Ontario K7L 3N6, Canada}
\email{\{piao, ksalomaa\}@cs.queensu.ca}
}
\def\titlerunning{Unranked Tree Automata}
\def\authorrunning{X. Piao \& K. Salomaa}


\maketitle

\begin{abstract}
    We consider  the state complexity of
basic operations on tree languages recognized by
deterministic unranked
     tree automata.
     For the operations of
     union and intersection the
upper and lower bounds of both weakly and strongly
 deterministic
     tree automata are obtained.
For tree concatenation we establish a tight upper bound that
is of a different order than the known state complexity
of concatenation of regular string languages.
We show that $(n+1) ( (m+1)2^n-2^{n-1} )-1$
     vertical states are sufficient, and
necessary in the worst case, to recognize the concatenation
of tree languages recognized by  (strongly or
weakly) deterministic automata with, respectively,
     $m$ and $n$ vertical states.
\end{abstract}

\section{Introduction}\label{s:in}

As XML \cite{dtddef} has  played increasingly important roles in
data representation and exchange through the web,  tree automata
have gained renewed interest, particularly tree automata operating
on unranked trees. XML documents can be abstracted as unranked
trees, which makes unranked tree automata a natural and
fundamental model for various XML processing tasks
\cite{CDG,MaN,Sc}. Both deterministic and nondeterministic
unranked tree automata have been studied.

One method to handle unranked trees is to encode them as ranked
trees and then use the classical
theory of ranked tree automata. However, the
encoding may result in trees of unbounded height since there is
no a priori restriction on the number of the children of a node in
unranked trees. Also depending on various applications, it may
be difficult to come up with a proper
 choice of the encoding method.

Descriptional complexity of finite automata and related
structures has been extensively studied
in recent years~\cite{GH,HK2,HK,Yu,Yu2}.
Here we consider operational state complexity of deterministic unranked
tree automata. Operational state complexity describes how the size
of an automaton varies under regularity preserving operations.
The corresponding results for string languages are well
known~\cite{galina,Yu,YuZhSa94},
however, very few results have been obtained for tree automata.
While   state complexity results for  tree automata
operating on ranked trees
 are often similar to corresponding
results on  regular string automata \cite{Yu}, the situation becomes
essentially different for automata operating on unranked trees.
An unranked tree automaton has two different types of states,
called horizontal and vertical states, respectively.
There are also other automaton models that
can be used to process  unranked trees, such
as nested word automata and stepwise tree automata. The state
complexity of these models has been studied
in~\cite{HS,mn,nestps}.

We study two different models of determinism for unranked
tree automata.
We call the usual deterministic unranked tree automaton \cite{CDG}
model where
the horizontal languages defining the transitions are specified by
DFAs (deterministic finite automata),
 a {\em weakly deterministic tree automaton} (or WDTA). For the
other variant of determinism for unranked tree automata,  we
refer to the corresponding automaton model as a
 {\em strongly deterministic unranked tree automaton} (or SDTA).
This model was introduced by Cristau, L\"oding and
Thomas~\cite{CLT}, see also
Raeymaekers and Bruynooghe~\cite{RB}.
SDTAs can be minimized efficiently and the
minimal automaton is unique~\cite{CLT}. On the other hand, the
minimization problem for WDTAs is NP-complete and the minimal
automaton need not be unique~\cite{mn}.

We give upper  and lower bounds for
the numbers of both vertical and
horizontal states for the operations of union and intersection.
The upper bounds for vertical states
are tight for both SDTAs and WDTAs. We also get upper bounds which
are almost tight for the number
of the horizontal states of SDTAs. Obtaining
a matching
lower bound for the horizontal states of WDTAs turns out to be
very problematic. This is mainly because the minimal WDTA may not
be unique and the minimization of WDTAs is intractable~\cite{mn}.
Also, the number of horizontal states of WDTAs can be reduced by
adding vertical states, i.e., there can be trade-offs between
the numbers of horizontal and vertical states, respectively.

The upper bounds for the number of
 vertical states for union
and intersection of WDTAs and SDTAs are,
as expected, similar to the upper
bound for the corresponding operation on ordinary string automata.
Already in the case of union and intersection,
the upper bounds for the numbers
horizontal states  are dramatically different for WDTAs and
SDTAs, respectively.
In an SDTA, the
horizontal language associated with label $\sigma$ is represented
with a single DFA $H_\sigma$ augmented with an output function
$\lambda$. The state assigned to a node labeled with $\sigma$ is
determined by the final state reached in $H_\sigma$ and $\lambda$.
On the other hand, in a WDTA, the horizontal languages associated
with a given label $\sigma$ and different states are represented
by distinct DFAs. The state assigned to a node labeled with $\sigma$
depends on the choice of the DFA.

We consider also the state complexity of
 (tree)
concatenation of  SDTAs. It is well known
that $m2^n-2^{n-1}$ states are sufficient to accept the
concatenation of an $m$ state DFA and an $n$ state
DFA~\cite{YuZhSa94}. However, the tight upper bound to accept the
concatenation of
unranked tree automata,
with $m$ and $n$ vertical states respectively,
turns out to be $(n+1) ( (m+1)2^n-2^{n-1} )-1$. The
factor $(n+1)$ is necessary here because the automaton accepting
the concatenation of two tree languages must keep track of the
computations where no concatenation has been done. For string
concatenation, there is only one path and the concatenation always
takes place  somewhere on that path.
For non-unary trees, there is no way
that the automaton can foretell on which branch the concatenation
is done and, consequently, the automaton for concatenation needs
considerably more states.
 It should be emphasized
that this phenomenon is not caused by any particular construction
used for the automaton to accept the concatenation of given tree
languages, and we have a matching lower bound result.

Since complementation is an ``easy'' operation for both strongly
and weakly deterministic tree automata, we do not investigate its
state complexity in this paper. Note that we do not require the
automaton models to be complete (i.e., some transitions may be
undefined). A (strongly or weakly) deterministic automaton
accepting the complement of a tree language recognized by the same
type of automaton would need at most one additional vertical state
and it is easy to see that this bound can be reached in the worst
case.

The  paper is organized as follows. Definitions of
unranked tree automata and other notations are given in
section~\ref{pre}. The upper bounds and corresponding
lower bounds  for union and
intersection of SDTAs are presented in section~\ref{sdta}. In
section~\ref{wdta}, the state complexity of union and
intersection of WDTAs is discussed. The tight bound for the
number of vertical
states for tree concatenation of SDTAs is
given in section~\ref{con}. The same construction works for
WDTAs.


\section{Preliminaries}\label{pre}

Here we briefly recall some notations and definitions concerning
trees and tree automata. A general reference on tree automata
is~\cite{CDG}.

Let $\bbbn$ be the set of non-negative integers. A {\em tree
domain} $D$ is a finite set of elements in $\bbbn^*$ with the
following two properties: (i) If $w\in D$ and $u$ is a prefix of
$w$ then $u\in D$. (ii) If $ui \in D$, $i\in \bbbn$ and $j<i$ then
$uj \in D$. The nodes in an unranked tree $t$ can be denoted by a
tree domain $dom(t)$, and $t$ is a mapping from $dom(t)$ to the
set of labels $\Sigma$. The set of $\Sigma$-labeled trees is
$T_\Sigma$.

For $t,t'\in T_\Sigma$ and $u \in dom(t')$, $t'(u\leftarrow t)$
denotes the tree obtained from $t'$ by replacing the subtree at
node $u$ by $t$. The concatenation of trees $t$ and $t'$ is
defined as
$t\cdot t' = \{t'(u\leftarrow t)\mid u\in
leaf(t')\}$.
The concatenation operation is extended in the
natural way to sets of trees $L_1$, $L_2$: $$L_1\cdot L_2 =
\bigcup_{t\in L_1, t'\in L_2} t\cdot t'.$$

%
We denote a tree $t = b ( a_1, \ldots ,a_n )$, whose root is
labeled by $b$ and leaves are labeled by $ a_1, \ldots ,a_n $,
simply as $b ( a_1 \ldots a_n)$. When $a_1= \ldots =a_n=a$, write
$t=b(a^n)$. By a slight abuse of notation, for a unary tree $t =
a_1 ( a_2 ( \ldots (a_n)\ldots)) $, we write $t = a_1 a_2 \ldots
a_n$ for abbreviation. When $a_1= \ldots =a_n=a$, we write $t=a^n$
for short. (In each case it should be clear from the context
whether $a^n$ refers to a sequence of leaves or to a unary tree.)


Next we briefly recall the definitions of the two variants of
deterministic bottom-up tree automata considered here.
A {\em weakly deterministic unranked tree automaton} (WDTA) is a
4-tuple $A=(Q,\Sigma,\delta,F)$ where $Q$ is a finite set of
states, $\Sigma$ is the alphabet, $F \subseteq Q$ is the set of
final states, $\delta$ is a mapping from $Q\times\Sigma$ to the
subsets of $(Q \cup \Sigma)^*$
which satisfies the condition that, for each $q
\in Q, \sigma \in \Sigma, \delta(q,\sigma)$ is a regular language
and for each label $\sigma$ and every two states $q_1\neq q_2$,
$\delta(q_1,\sigma)\bigcap\delta(q_2,\sigma)=\emptyset$. The
language $\delta(q,\sigma)$ is called the {\em horizontal
language} associated with $q$ and $\sigma$ and it is specified
by a DFA $H_{q,\sigma}^A$.

Roughly speaking, a WDTA operates as follows.
If $A$ has reached the children of a
$\sigma$-labelled node $u$ in states
 $q_1$, $q_2$ ,..., $q_n$, the
computation assigns state $q$ to node $u$ provided that
$q_1q_2...q_n\in\delta(q,\sigma)$. In the sequence
$q_1q_2...q_n$ an element $q_i \in \Sigma$ is interpreted to
correspond to a leaf labeled by that symbol.
A WDTA is a deterministic
hedge automaton~\cite{CDG} where each horizontal language
is specified using a DFA.

Note that in the usual definition of~\cite{CDG} the
horizontal languages
are subsets of $Q^*$. In order to simplify some constructions,
 we allow also the use of symbols  of the
alphabet $\Sigma$ in the horizontal languages, where
a symbol $\sigma \in \Sigma$ occurring in a word of
a horizontal language is always interpreted to label
a leaf of the tree.
The convention does not
change the state complexity bounds in any
significant way because we use small
constant size alphabets and we can think that the tree automaton
assigns to each leaf labeled by $\sigma \in \Sigma$ a particular
state that is not used anywhere else in the computation.

A {\em strongly deterministic unranked tree automaton} (SDTA) is a
4-tuple $A=(Q,\Sigma,F,\delta)$, where $Q, \Sigma, F$ are
similarly defined as for WDTAs. For each $a \in \Sigma$, the
horizontal languages $\delta(q, a)$, $q \in Q$, are defined by a
single DFA augmented with an output function as follows. For $a
\in \Sigma$ define $D_a=(S_a,Q \cup \Sigma,s_a^0,\gamma_a,E_a,
\lambda_a)$ where $(S_a,Q \cup \Sigma,s_a^0,\gamma_a,E_a)$ is a
DFA and $\lambda_a$ is a mapping $S_a\rightarrow Q$. For all $q
\in Q$ and $a \in \Sigma$, the horizontal language $\delta(q,a)$
is specified by $D_a$ as the set $\{ w \in (Q \cup \Sigma)^* \mid
\lambda_a(\gamma_a^*(s_a^0,w))=q\}$. Intuitively, when $A$ has
reached the children of a node $u$ labelled by $a$ in states $q_1,
\ldots, q_m$ (an element $q_i \in \Sigma$ is interpreted as a
label of a leaf node), the state at $u$ is determined (via the
function $\lambda_a$) by the state that the DFA $D_a$ reaches
after reading the word $q_1 \cdots q_m$. More information on
SDTA's can be found in~\cite{CLT}.

Given a tree automaton $A=(Q,\Sigma,F,\delta)$, the states in $Q$
are called {\em vertical states\/}. The DFAs recognizing
the horizontal languages are called {\em horizontal
DFAs\/} and their states
 are called horizontal states.
We define the {\em (state) size of $A$,}  ${\rm
size}(A)$, as a pair of integers $[ |Q|, n ]$, where
$n$ is the sum of the sizes of all horizontal DFAs associated
with $A$.

\section{Union and intersection}

We investigate the state complexity of union and intersection
operations on unranked tree automata. The upper bounds on the
numbers of vertical states are similar for SDTAs and WDTAs,
however the upper bounds on the numbers of horizontal states
differ between the two models.

\subsection{Strongly deterministic tree automata}\label{sdta}

The following result gives the upper bounds and the lower bounds
for the operations of union and intersection for SDTAs.

\begin{theorem}\label{xxx}
For any two arbitrary SDTAs $A_i=(Q_i,\Sigma,\delta_i,F_i)$,
$i=1,2$, whose transition function associated with $\sigma$ is
represented by a DFA $H_{\sigma}^{A_i}=(C_{\sigma}^i, Q_i \cup
\Sigma, \gamma_{\sigma}^i, c_{\sigma,0}^i, E_{\sigma}^i)$, we have
\begin{description}
  \item[1] Any SDTA $B_\cup$ recognizing $L(A_1)\cup L(A_2)$ satisfies that $${\rm size}(B_{\cup}) \leq [ \; (|Q_1|+1)\times
(|Q_2|+1)-1; \; \sum_{\sigma \in \Sigma} ((|C_{\sigma}^{1}|+1)
\times (|C_{\sigma}^{2}|+ 1) -1) \; ].$$
  \item[2] Any SDTA $B_\cap$ recognizing $L(A_1)\cap L(A_2)$ satisfies that $${\rm size}(B_{\cap}) \leq [ \; |Q_1|\times
|Q_2|; \; \sum_{\sigma \in \Sigma} |C_{\sigma}^{1}| \times
|C_{\sigma}^{2}| \; ].$$
  \item[3] For integers $m, n\geq 1$ and relatively prime numbers
  $k_1,k_2,\ldots,k_m,k_{m+1},\ldots,\\ k_{m+n}$, there exists
  tree languages $T_1$ and $T_2$ such that $T_1$ and $T_2$, respectively, can be
  recognized by SDTAs with $m$ and $n$ vertical states,
  $\prod_{i=1}^m k_i+O(m)$ and $\prod_{i=1+m}^{m+n} k_i+O(n)$ horizontal states, and
\begin{description}
  \item[i] any SDTA recognizing $T_1\cup T_2$ has at least
 $(m+1)(n+1)-1$ vertical
  states and $\prod_{i=1}^{m+n} k_i$ horizontal states.
  \item[ii] any SDTA recognizing $T_1\cap T_2$ has at least $mn$ vertical
  states and $\prod_{i=1}^{m+n} k_i$ horizontal states.
\end{description}
\end{description}
\end{theorem}


The upper bounds on vertical and horizontal states are obtained
from product constructions, and Theorem~\ref{xxx} shows that for
the operations of union and intersection on SDTAs the upper bounds
are tight for vertical states and almost tight for horizontal
states.

\subsection{Weakly deterministic automata}\label{wdta}

In this section, the upper bounds on the numbers of vertical and
horizontal states for the operations of union and intersection on
WDTAs are investigated, and followed by matching lower bounds on
the numbers of vertical states.

\begin{lemma}\label{union}
Given two WDTAs $A_i=(Q_i,\Sigma,\delta_i,F_i)$, $i=1,2$, each
horizontal language $\delta_i(q,\sigma)$ is represented by a DFA
$D_{q,\sigma}^{A_i}=(C_{q,\sigma}^i, Q_i \cup \Sigma,
\gamma_{q,\sigma}^i, c_{q,\sigma,0}^i, E_{q,\sigma}^i)$.

The language $L(A_1)\cup L(A_2)$ can be recognized by a WDTA
$B_{\cup}$ with
\begin{eqnarray*} \ & {\rm size}(B_{\cup}) \leq [
\; (|Q_1|+1)\times (|Q_2|+1)-1;
&\ \\
\ & |\Sigma| \times (\sum_{q \in Q_1,p \in
Q_2}|D_{q,\sigma}^{A_1}| \times |D_{p,\sigma}^{A_2}| + \sum_{q\in
Q_1}|D_{q,\sigma}^{A_1}|\times \prod_{p \in
Q_2}|D_{p,\sigma}^{A_2}| + \sum_{p\in Q_2}|D_{p,\sigma}^{A_2}|
&\
\\
\ & \times \prod_{q \in Q_1}|D_{q,\sigma}^{A_1}|) \; ] &
\end{eqnarray*}

The language $L(A_1)\cap L(A_2)$ can be recognized by a WDTA
$B_{\cap}$ with \begin{eqnarray*}\ & {\rm size}(B_{\cap}) \leq [
\; |Q_1|\times |Q_2|; \; |\Sigma| \times \sum_{q \in Q_1,p\in Q_2}
|D_{q,\sigma}^{A_1}| \times |D_{p,\sigma}^{A_2}| \; ].&\
\end{eqnarray*}
\end{lemma}
The theorem below shows that the upper bounds for the vertical
states are tight.

\begin{theorem}\label{dtadfa}
For any two WDTAs $A_1$ and $A_2$ with $m$ and $n$ vertical states
respectively, we have
\begin{itemize}
  \item[1] any WDTA recognizing $L(A_1)\cup L(A_2)$ needs at
  most $(m+1)(n+1)-1$ vertical states,
  \item[2] any WDTA recognizing $L(A_1)\cap L(A_2)$ needs at
  most $mn$ vertical states,
  \item[3] for any integers $m,n\geq 1$, there exist tree languages $T_1$
  and $T_2$ such that $T_1$ and $T_2$ can be recognized by
  WDTAs with $m$ and $n$ vertical states respectively, and any WDTA
  recognizing $T_1 \cup T_2$ has at least $(m+1)(n+1)-1$ vertical
  states, and any WDTA recognizing $T_1 \cap T_2$ has at least $mn$ vertical
  states.
\end{itemize}
\end{theorem}


\begin{problem}
Are the upper bounds for the numbers of horizontal states given in
Lemma~\ref{union} tight?
\end{problem}
In the case of WDTAs we do not have
a general method 
to establish lower bounds on the number of the horizontal states.
It remains an open question to give (reasonably) tight lower
bounds on the number of horizontal states needed to recognize the
union or intersection of tree languages recognized by two WDTA's.
\section{Concatenation of strongly deterministic tree automata}\label{con}

We begin by  giving a construction of an SDTA recognizing the
concatenation of two tree languages recognized by given
 SDTAs.

\begin{lemma}\label{cons}
Let $A_1$ and $A_2$ be two arbitrary SDTAs.
$A_i=(Q_i,\Sigma,\delta_i,F_i)$, $i=1,2$, transition function for
each $\sigma\in\Sigma$ is represented by a DFA
$H_{\sigma}^{A_i}=(C_{\sigma}^i, Q_i \cup \Sigma,
\gamma_{\sigma}^i, c_{\sigma,0}^i, E_{\sigma}^i)$ with an output
function $\lambda_\sigma^i$.

The language $L(A_2)\cdot L(A_1)$ can be recognized by an SDTA $B$
with
$${\rm size}(B ) \leq [ \; (|Q_1| + 1)\times (2^{|Q_1|}
\times (|Q_2|+1)-2^{|Q_1|-1})-1; \; |\Sigma| (|C_{\sigma}^2|+1)
(|C_{\sigma}^1|+1)\times 2^{|C_{\sigma}^1| + 1} \; ].$$
\end{lemma}

\begin{proof}
Choose $B = (Q_1' \times Q_1'' \times Q_2', \Sigma, \delta, F)$,
where $Q_1'= Q_1 \cup \{dead\} $, $Q_1''= {\cal P} (Q_1) $, $Q_2'
= Q_2 \cup \{ dead\}$. Let $P_2 \subseteq Q_1 $. $( p_1 , P_2
,q)\in Q_1' \times Q_1'' \times Q_2'$ is final if there exists
$p\in P_2$ such that $p\in F_1$.

The transition function $\delta$ associated with each $\sigma$ is
represented by a DFA $H_{\sigma}^{B}=(  S \times S'' \times S',
(Q_1' \times Q_1'' \times Q_2') \cup \Sigma, \mu, ( c_{\sigma,0}^1
, ( \{ c_{\sigma,0}^1 \} , 0) ,c_{\sigma,0}^2), V)$ with an output
function $\lambda_\sigma^{B}$, where $S = C_{\sigma}^1 \cup
\{dead\} $, $S'' =  {\cal P} ( C_{\sigma}^1 ) \times \{ 0 ,1 \} $,
$S' = C_{\sigma}^2 \cup \{dead\} $. Let $C_2\subseteq
C_{\sigma}^1$, $x=1,0$. $( c_1 ,( C_2 , x ) , c^2) \in S \times
S'' \times S'$ is final if $c^2 \in E_{\sigma}^2$ or there exists
$c \in c_1 \cup C_2$ such that $c \in E_{\sigma}^1$. $\mu$ is
defined as below:

For any input $a\in \Sigma$,
$$
\mu(( c_1 , ( C_2 , x) ,c^2), a) = ( \gamma_\sigma^1(c_1, a),
(\bigcup_{c_2\in C_2} \gamma_\sigma^1(c_2, a), x),
\gamma_\sigma^2(c^2, a))
$$
For any input $( p_1 , P_2  ,q)\in Q_1' \times Q_1'' \times Q_2'
$, if $P_2 \neq \emptyset$,
$$
\mu(( c_1 , ( C_2 , 0) , c^2), ( p_1 , P_2  ,q) ) = (
\gamma_\sigma^1(c_1, p_1), (\bigcup_{p_2\in P_2}
\gamma_\sigma^1(c_1, p_2) , 1), \gamma_\sigma^2(c^2, q))
$$
$$
\mu((  c_1 , ( C_2 , 1) , c^2), ( p_1 , P_2  ,q) ) = (
\gamma_\sigma^1(c_1, p_1), ( \bigcup_{p_2\in
P_2}\gamma_\sigma^1(c_1, p_2) \cup \bigcup_{c_2\in C_2}
\gamma_\sigma^1(c_2, p_1) , 1), \gamma_\sigma^2(c^2, q))
$$
if $P_2 = \emptyset$,
$$
\mu(( c_1 , ( C_2 , 0) , c^2), ( p_1 , \emptyset  ,q) ) = (
\gamma_\sigma^1(c_1, p_1), ( \emptyset , 0), \gamma_\sigma^2(c^2,
q))
$$
$$
\mu((  c_1 , ( C_2 , 1) , c^2), ( p_1 , \emptyset  ,q) ) = (
\gamma_\sigma^1(c_1, p_1), (\bigcup_{c_2\in C_2}
\gamma_\sigma^1(c_2, p_1) , 1), \gamma_\sigma^2(c^2, q))
$$
Write the computation above in an abbreviated form as $\mu((  c_1
, ( C_2 , x) , c^2), r ) = ( p_1' , P_2'  ,q' )$, $r\in \Sigma\cup
Q_1'\times Q_1''\times Q_2'$. When compute $p_1'$ and $q'$, if any
$\gamma_\sigma^i (c, \alpha)$, $i=1,2$, $c=c_1,c^2$, $\alpha \in
\Sigma \cup Q_i $, is not defined in $A_i$, assign $dead$ to
$p_1'$ or $q'$. When compute $P_2'$, add nothing to $P_2'$ if any
$\gamma_\sigma^i (c, \alpha)$ is not defined.

Let $p_{leaf}\in Q_1$ denote the state assigned to the leaf in
$A_1$ substituted by a tree in $L(A_2)$. $\lambda_\sigma^{B}$ is
defined as: for any final state $e = ( c_1 , ( C_2 , x) ,c^2)$,
$x_1 = c_1 \cap E_\sigma^1$, $X_2 = C_2 \cap E_\sigma^1$,
\begin{itemize}
  \item[1] If $c^2 \in E_\sigma^2$
$$
\lambda_\sigma^{B}(e) = \left\{
\begin{array}{l}
 ( \lambda_\sigma^1 (x_1) ,  p_{leaf} \cup \bigcup_{x_2\in X_2}\lambda_\sigma^1 (x_2) , \lambda_\sigma^2 (c^2) ), \mbox{ if }
\lambda_\sigma^2 (c^2) \in F_2 \mbox{ and } x=1
 \\
  ( \lambda_\sigma^1 (x_1) ,  p_{leaf} , \lambda_\sigma^2 (c^2) ),
\mbox{ if } \lambda_\sigma^2 (c^2) \in F_2 \mbox{ and } x=0
 \\
(\lambda_\sigma^1 (x_1) , \bigcup_{x_2\in X_2} \lambda_\sigma^1
(x_2) ,\lambda_\sigma^2 (c^2) ), \mbox{ if } \lambda_\sigma^2
(c^2) \notin F_2 \mbox{ and } x=1
 \\
 ( \lambda_\sigma^1 (x_1) , \emptyset , \lambda_\sigma^2 (c^2) ), \mbox{
if } \lambda_\sigma^2 (c^2) \notin F_2 \mbox{ and } x=0
\end{array} \right.
$$
\item[2] If 
$c^2 \notin E_{\sigma}^2$,
$$
\lambda_\sigma^{B}(e) = \left\{
\begin{array}{l}
( \lambda_\sigma^1 (x_1) , \emptyset , dead ) \mbox{ if } x = 0
 \\
( \lambda_\sigma^1 (x_1) , \bigcup_{x_2\in X_2} \lambda_\sigma^1
(x_2) ,dead ) \mbox{ if } x=1
\end{array} \right.
$$
If $x_1 = \emptyset$, define $ \lambda_\sigma^1 (x_1) = dead $. If
$X_2 = \emptyset$, define $\bigcup_{x_2\in X_2} \lambda_\sigma^1
(x_2) = \emptyset$.
\end{itemize}
The state in $B$ has three components $( p_1 , P_2 ,q)$. $p_1$ is
used to keep track of $A_1$'s computation where no concatenation
is done. $p_1$ is computed by the first component $c_1$ in the
state of $H_{\sigma}^{B}$. $P_2$ traces the computation where the
concatenation takes place. In a state $( c_1 , ( C_2 , x) ,c^2)$
of $H_{\sigma}^{B}$, $x=1$ (or $x=0$) records there is (or is not)
a concatenation in the computation. The third component $q$ keeps
track of the computation of $A_2$. When a final state is reached
in $A_2$, which means a concatenation might take place, an initial
state $p_{leaf}$ is added to $P_2$, which is achieved by the
$\lambda_\sigma^{B}$ function in $B$.

According to the definition of $\lambda_\sigma^{B}$, when
$\lambda_\sigma^2 (c^2) \in F_2$, $p_{leaf}$ is always in the
second component of the state. Exclude the cases when
$\lambda_\sigma^2 (c^2) \in F_2$, and $p_{leaf}$ is not in the
second component of the state, and we do not require $B$ be
complete. $B$ has $(|Q_1| + 1)\times (2^{|Q_1|} \times
(|Q_2|+1)-2^{|Q_1|-1})-1$ vertical states in worst case.
\end{proof}
\endpf

Lemma~\ref{cons} gives an upper bound on both the numbers of
vertical and horizontal states recognizing the concatenation of
$L(A_2)$ and $L(A_1)$. In the following we give a matching lower
bound for the number of vertical states of any SDTA recognizing
$L(A_2) \cdot L(A_1)$.




For our lower bound construction we define tree languages
consisting of trees where, roughly speaking, each branch
belongs to the worst-case languages used for string concatenation
in~\cite{YuZhSa94} and, furthermore, the minimal DFA
reaches the same state at an arbitrary node $u$ in computations starting
from any two leaves below $u$. For technical reasons, all leaves
of the trees are labeled by a fixed symbol and the strings used to
define the tree language do not include the leaf symbols.

As shown in
Figure~\ref{f:dfa}, $A$ and $B$ are the DFAs used in Theorem~1 of
\cite{YuZhSa94} except that a self-loop labeled by
an additional symbol $d$ is added to
each state in $B$.
We use the symbol $d$ as an identifier of DFA $B$,
which always leads  to a dead state in the computations of $A$.
This will be useful for establishing that all vertical states
of the SDTA constructed as in Lemma~\ref{cons} are needed
to recognize the concatenation of tree languages defined below.

\begin{figure}
\centering
\includegraphics[height=3cm]{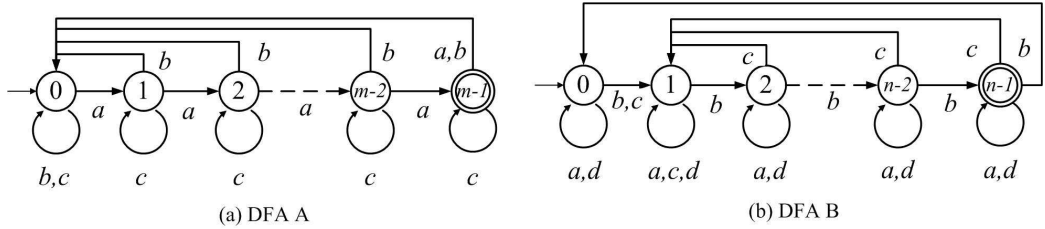}
\caption{DFA $A$ and $B$}\label{f:dfa}
\end{figure}

Based on the DFAs $A$ and $B$ we define the tree languages
$T_A$ and $T_B$ used in our lower bound construction.
The tree language $T_B$ consists of $\Sigma$-labeled trees $t$,
$\Sigma = \{ a, b, c, d \}$, where:
\begin{enumerate}
\item All leaves are labeled
by $a$ and if a node $u$ has a child that is a leaf, then all the
children of $u$ are leaves.
\item $B$ accepts the string of symbols labeling
a path from any node of height one to the root.
\item The following
holds for any $u \in {\rm dom}(t)$ and any nodes
 $v_1$ and $v_2$
 of height one below $u$. If $w_i$ is the string
of symbols labeling the path from $v_i$ to $u$, $i = 1, 2$, then
$B$ reaches the same state after reading strings $w_1$ and $w_2$.
\end{enumerate}
Intuitively, the above condition means that when, on a tree of
$T_B$, the DFA $B$ reads strings of symbols labeling paths
starting from nodes of height one upwards, the computations
corresponding to different paths ``agree'' at each node. This
property is used in the construction of an SDTA $M_B$ for $T_B$
below.

Note that the computations of $B$ above are started from the nodes
of height one and they ignore the leaf symbols. This is done for
technical reasons because in tree concatenation a leaf symbol is
replaced by a tree, i.e., the original symbol labeling the leaf
will not appear in the resulting tree.

$T_B$ can be recognized by an SDTA $M_B=(Q_B,\{a,b,c,d\},
\delta_B, F_B)$ where $Q_B=\{0,1,\ldots,n-1\}$ and $F_B=\{n-1\}$.
The transition function is defined as:
\begin{itemize}
  \item[(1)] $\delta_B(0,a)=\epsilon$,
  \item[(2)] $\delta_B(i,a)=\bigcup_{0 \leq i \leq n-1} i^+$,
  \item[(3)] $\delta_B(i,d)=\bigcup_{0 \leq i \leq n-1} i^+$,
  \item[(4)] $\delta_B(j, b) = (j-1)^+, 1 \leq j \leq n-1$ and
  $\delta_B(0, b) = (n-1)^+$,
  \item[(5)] $\delta_B(1,c)=\{0, \ldots, n-1\}^+$.
\end{itemize}

The tree language $T_A$ and an SDTA $M_A$ recognizing it are
defined similarly based on the DFA $A$. Note that $T_A$ has no
occurrences of the symbol $d$ and $M_A$ has no transitions defined
on $d$.
The SDTAs $M_A$ and $M_B$ have $m$ and $n$
vertical states, respectively.

An SDTA $C$ recognizing tree language $T_A \cdot T_B$\ \
\footnote{Recall from section~\ref{pre} that $T_A \cdot T_B$
consists of trees where in some tree of $T_B$ a leaf is replaced
by a tree of $T_A$.} is obtained from $M_A$ and $M_B$ using the
construction given in Lemma~\ref{cons}. The vertical states in $C$
are of the following form
\begin{equation}\label{state1}
(q,S,p), 0\leq q\leq n, S\subseteq \{0,1,\ldots,n-1\}, 0\leq p\leq
m,
\end{equation} where if $p=m-1$ then $0\in S$, and if $S=\emptyset$
then $q=n$ and $p=m$ can not both be true. The number of states
in~(\ref{state1}) is $(n+1) ((m+1)2^n-2^{n-1})-1$. State $q=n$ (or
$p=m$) denotes $q=dead$ (or $p=dead$) in the construction of
lemma~\ref{cons}. We will show that $C$ needs at least $(n+1)
((m+1)2^n-2^{n-1})-1$ vertical states. We prove this by showing
that each state in~(\ref{state1}) is reachable and all states are
pairwise inequivalent, or distinguishable. Here distinguishability
means that for any distinct states $q_1$ and $q_2$ there exists $t
\in T_\Sigma[x]$ such that the (unique deterministic) computation
of $C$ on $t(x \leftarrow q_1)$ leads to acceptance if and only if
the computation of $C$ on $t(x \leftarrow q_2)$ does not lead to
acceptance.

\begin{lemma}
All states of $C$ are reachable.
\label{reach}
\end{lemma}

\begin{proof}
We introduce the following notation. For a unary tree
\\$t=a_1(a_2(\ldots a_m(b)\ldots))$, we denote
$word(t)=a_ma_{m-1}\ldots a_1 \in \Sigma^*$. Note that $word(t)$
consists of the sequence of labels of $t$ from the node of height
one to the root, and the label of the leaf is not included.

We show that all the states in~(\ref{state1}) are reachable
by using induction on $|S|$.

When $|S|=0$, $(i, \emptyset, j)$, $0\leq i\leq n-1$, $0\leq j\leq
m-2$ is reachable from $(0,\emptyset,0)$ by reading tree $t$ where
$word(t)=b^ia^j$. State $(n, \emptyset, j)$, $1\leq j\leq m-2$ is
reachable from $(0,\emptyset,0)$ by reading tree $a(t_1,t_2)$
where $word(t_1)=ba^{j-1}$ and $word(t_2)=b^2a^{j-1}$. State
$(n,\emptyset,0)$ is reachable by reading symbol $b$ from state
$(n, \emptyset, j)$, $1\leq j\leq m-2$. State $(i, \emptyset, m)$,
$0\leq i\leq n-1$ is reachable from $(0,\emptyset,0)$ by reading
tree $b(t_1,t_2)$ where $word(t_1)=b^{i-1}a$ and
$word(t_2)=b^{i-1}a^2$.

When $|S|=1$, $(i, \{0\}, m-1)$, $0\leq i\leq n-1$ is reachable
from $(0,\emptyset,0)$ by reading tree $t$ where
$word(t)=b^ia^{m-1}$.

State $(n, \{0\}, m-1)$, is reachable from $(0,\emptyset,0)$ by
reading tree $a(t_1,t_2)$ where $word(t_1)=ba^{m-2}$ and
$word(t_2)=b^2a^{m-2}$.

State $(i, \{0\}, j)$, $0\leq i\leq n$, $0\leq j\leq m-2$ is
reachable from $(i, \{0\}, m-1)$ by reading a sequence of unary
symbol $a^{1+j}$.

State $(i, \{0\}, m)$, $0\leq i\leq n-1$ is reachable from
$(0,\emptyset,0)$ by reading tree $t$ where
$word(t)=b^{i}a^{m-1}d$.

From $(0,\emptyset,0)$ by reading subtree $b(b(a),b(b(a)))$, state
$(n,\emptyset,0)$ is reached. State $(n, \{0\}, m)$ is reached
from $(n,\emptyset,0)$ by reading a sequence of unary symbols
$a^{m-1}d$.

That is all the states $(i, \{0\}, j)$, $0\leq i\leq n$, $0\leq
j\leq m$ are reachable.

Then state $(i, \{k\}, j)$, $0\leq i\leq n-1$, $0\leq j\leq m-1$,
$1\leq k\leq n-1$ is reachable from $(\overline{i-1},\{k-1\},j)$
by reading a sequence of unary symbols $ba^j$. For any integer
$x$,
$$
\overline{x} = \left\{
\begin{array}{l}
x \mbox{ if } x \geq 0
 \\
n+x \mbox{ if } x< 0
\end{array} \right.
$$
State $(n, \{k\}, j)$, $0\leq j\leq m-1$, $1\leq k\leq n-1$ is
reachable from $(n,\{k-1\},j)$ by reading a sequence of unary
symbols $ba^j$. State $(i, \{k\}, m)$, $0\leq i\leq n-1$, $1\leq
k\leq n-1$ is reachable from $(\overline{i-1},\{k-1\},m)$ by
reading a unary symbol $b$. State $(n, \{k\}, m)$, $1\leq k\leq
n-1$ is reachable from $(n,\{k-1\},m)$ by reading a unary symbol
$b$.

That is all the states $(i, \{k\}, j)$, $0\leq i\leq n$, $0\leq
j\leq m$, $0\leq k\leq n-1$ are reachable.

Now assume that for $|S|\leq z$, all the states $(i, S, j)$,
$0\leq i\leq n$, $0\leq j\leq m$, $S\subseteq \{0,\ldots,n-1\}$
are reachable. And this is the inductive assumption.

We will show that any state $(x, S', y)$, $0\leq x\leq n$, $0\leq
y\leq m$, $|S'|=z+1$ is reachable.

First consider the case where $y\neq m-1$. Let $s_1 > s_2 > \ldots
> s_z > s_{z+1}$ be the elements in $S'$. Let $P= \{s_1-s_{z+1} ,
s_2-s_{z+1} , \ldots , s_z-s_{z+1}\}$.

When $0\leq x\leq n-1$, according to the inductive assumption,
state $(\overline{x-s_{z+1}}, P , 0)$, is reachable. Then state
$(\overline{x-s_{z+1}}, P\cup \{0\}, m-1)$ is reachable from
$(\overline{x-s_{z+1}}, P, 0)$ by reading a sequence of unary
symbols $a^{m-1}$. State $(x, S', y)$, $0\leq y\leq m-2$ is
reachable from $(\overline{x-s_{z+1}}, P\cup \{0\}, m-1)$ by
reading a sequence of unary symbols $b^{s_{z+1}}a^{y}$. State $(x,
S', m)$ is reachable from $(\overline{x-s_{z+1}}, P\cup \{0\},
m-1)$ by reading a sequence of unary symbols $b^{s_{z+1}}d$.

When $x=n$, according to the inductive assumption, state $(n, P ,
0)$, is reachable. Then state $(n, P\cup \{0\}, m-1)$ is reachable
from $(n, P, 0)$ by reading a sequence of unary symbols $a^{m-1}$.
$(n, S', y)$, $0\leq y\leq m-2$ is reachable from $(n, P\cup
\{0\}, m-1)$ by reading a sequence of unary symbols
$b^{s_{z+1}}a^{y}$. State $(n, S', m)$ is reachable from $(n,
P\cup \{0\}, m-1)$ by reading a sequence of unary symbols
$b^{s_{z+1}}d$.

Now consider the case when $y=m-1$. According to the definition of
(\ref{state1}), $0\in S'$. According to the inductive assumption,
state $(x, S'-\{0\}, m-2)$ is reachable. Then state $(x, S', m-1)$
is reachable by reading a unary symbol $a$.

Since $(x, S', y)$ is an arbitrary state with $|S'|=z+1$, we have
proved that all the states $(x, S', y)$, $0\leq x\leq n$, $0\leq
y\leq m$, $|S'|=z+1$ is reachable.

Thus, all the states in (\ref{state1}) are reachable.
\end{proof}\endpf

\begin{lemma}
All states of $C$ are pairwise inequivalent. \footnote{Proof
omitted due to length restriction.} \label{inequi}
\end{lemma}

According to the upper bound in Lemma~\ref{cons} and
Lemmas~\ref{reach} and~\ref{inequi}, we have proved the following
theorem.

\begin{theorem}\label{ti}
For arbitrary SDTAs $A_1$ and $A_2$, where
$A_i=(Q_i,\Sigma,\delta_i,F_i)$, $i=1,2$, any SDTA
$B=(Q,\Sigma,\delta,F)$ recognizing $L(A_2)\cdot L(A_1)$ satisfies
$|Q|\leq (|Q_1| + 1)\times (2^{|Q_1|} \times
(|Q_2|+1)-2^{|Q_1|-1})-1$.

For any integers $m,n\geq 1$, there exists tree languages $T_A$ and
$T_B$, such that $T_A$ and $T_B$ can be recognized by  SDTAs
having $m$ and $n$ vertical states, respectively,
 and any SDTA recognizing
$T_A \cdot T_B$ needs at least $(n+1) ( (m+1)2^n-2^{n-1} )-1$
vertical states.
\end{theorem}

We do not have a matching lower bound for the number of horizontal
states given by Lemma~\ref{cons}. With regards to the number of
vertical states, both the upper bound of Lemma~\ref{cons} and the
lower bound of Theorem~\ref{ti} can be immediately modified for
WDTAs. (The proof holds almost word for word.)  In the case of
WDTAs, getting a good lower bound for the number of horizontal
states would likely be very hard.

\section{Conclusion}

We have studied the operational state complexity of two variants
of deterministic unranked tree automata. For  union and
intersection, tight upper bounds on the number of vertical states
were established for both strongly and weakly deterministic
automata. An almost tight upper bound on the number of horizontal
states was obtained in the case of strongly deterministic unranked
tree automata. For weakly deterministic  automata,
lower bounds on the numbers of horizontal states are hard to
establish because there can be trade-offs between the numbers of
vertical and horizontal states
.
This is indicated also by the fact that minimization of weakly
deterministic unranked tree automata is intractable and the
minimal automaton need not be unique~\cite{mn}.

As  ordinary
strings can be viewed as unary trees, it is easy to predict that
the state complexity of a given operation for tree automata should
be greater or equal to  the state complexity of the corresponding
operation on string languages.
As our main result, we showed that for
deterministic unranked tree automata, the
 state complexity of concatenation of an $m$ state and
an $n$ state automaton is at most $(n+1) ( (m+1)2^n-2^{n-1} )-1$
and that this bound can be reached in the worst case.
The bound is
of a different order than the known state complexity $m2^n-2^{n-1}$
of  concatenation of regular string languages.

\end{document}